\documentclass[12]{amsart}                                            
\usepackage{amsmath,amssymb,amsthm,comment,mathtools}
\usepackage[usenames]{color}
\usepackage{hyperref}
\usepackage{url}

\newtheorem{theorem}{Theorem}

\newtheorem*{theorem*}{Theorem}
\newtheorem*{claim*}{Claim}
\newtheorem*{corollary*}{Corollary}
\newtheorem*{lemma*}{Lemma}
 
\theoremstyle{definition}

\newtheorem{definition}[theorem]{Definition}
\newtheorem*{definition*}{Definition}
 
\newtheorem*{example*}{Example}
\newtheorem*{experiment*}{Experiment}
\newtheorem*{notation*}{Notation}
\newtheorem{proviso}{Proviso}
\newtheorem{problem}{Problem}
\newtheorem*{remark*}{Remark}
\newtheorem*{terminology*}{Terminology}

\newcommand{\A}{\smallskip\noindent{\tt A: }}
\newcommand{\BA}{\mathrm{BA}}

\newcommand{\C}{\mathbb C}
\newcommand{\CO}{\ensuremath{\mathrm{CO}}}

\newcommand\given[1][]{\:#1\vert\:}

\renewcommand{\H}{\ensuremath{\mathcal H}}
\newcommand{\ket}[1]{\ensuremath{|#1\rangle}}

\newcommand{\Ob}{\text{Ob}}

\renewcommand{\P}{\mathcal P}
\renewcommand\phi{\varphi}

\newcommand{\Q}{\smallskip\noindent{\tt Q: }}

\newcommand{\cS}{\ensuremath{\mathcal S}}

\title{Negative probabilities, II\\ 
What they are and what they are for
}
\author{Andreas Blass}
\address{Mathematics Department\\
University of Michigan\\
Ann Arbor, MI 48109--1043, U.S.A.}
\email{ablass@umich.edu}
\author{Yuri Gurevich}
\address{Computer Science and Engineering\\
University of Michigan\\
Ann Arbor, MI  48109-2121, U.S.A}
\email{gurevich@umich.edu}

\begin{document}
\maketitle

\begin{quote}\raggedleft\small\it
What are numbers and what are they for?
\footnote{Was sind und was sollen die Zahlen?}\\
--- Richard Dedekind, 1888
\end{quote}

\begin{abstract}
A signed probability distribution may extend a given traditional probability
from observable events to all events. We formalize and illustrate this approach. We also illustrate its limitation. We argue that the right question is not what negative probabilities are but what they are for. 
\end{abstract}

\section{Introduction}
\label{sec:intro}

The idea of negative probabilities arose in quantum mechanics \cite{Wigner,Moyal, Dirac,Feynman}. This is not surprising. The weirdness of quantum mechanics required bold fresh ideas. The physicists have been using signed probability distributions primarily in connection with phase spaces for quantum systems \cite{Zachos} as suggested by Wigner in his 1932 paper \cite{Wigner}. 

The present article was provoked by the paper ``An operational interpretation of negative probabilities and no-signaling models'' by Samson Abramsky and Adam Brandenburger \cite{AB}. It occurred to us that the right question may not be how to interpret negative probabilities but how to employ them. 

Not everybody thinks that the idea of negative probabilities is a good one. Let's hear a critic. In 1935, Albert Einstein, Boris Podolsky and Nathan Rosen surmised that quantum mechanics is incomplete, i.e., that some hidden variables are missing \cite{EPR}. John Bell proved that local hidden-variable theories contradict quantum-mechanical statistics \cite{Bell}. It seems natural to try to save the hidden-variable approach by means of negative probabilities, and here is what our critic, Itamar Pitowsky, has to say about that  \cite[page~148]{Pitowsky}.

\begin{quote}
What makes the classical hidden variable theories ``classical'' is the
identification of ``mixtures of hidden variable states'' as probability measures\dots The logical step to take \dots is \dots to use (cheap) tricks such as negative ``probability'', or even complex ``probability'' values. Formally we may be able to ``solve'' our problem, but then the term ``probability'' loses completely its meaning. \dots It is absurd to talk about an urn containing $-17$ red balls or $3e^{i\pi/12}$ wooden balls.
\end{quote}

We wrote about negative probabilities once \cite{G224}. 
Here we wish to defend the cheap trick of negative probabilities.
Indeed, the standard frequential interpretation of probabilities does
not apply to negative probabilities. But here is an appetizing
analogy. For a long time, the standard interpretation of numbers was
quantity. It is absurd to talk about the quantity of
$3e^{i\pi/12}$. Formally the complex numbers allow us to ``solve'' equations like $x^2=-1$ but, one might argue, the term ``number'' loses its meaning. 

The mathematical trick of introducing complex numbers paid off richly. Some hard number-theoretic problems have been solved using the methods of complex analysis. Eventually it even became possible to give a physical meaning to complex numbers, e.g., a complex number $a+bi$ can be interpreted as impedance where the real part $a$ is resistance and the imaginary part $b$ is the reactance \cite{W-EI}, and complex numbers of absolute value $\le1$ can be interpreted as quantum amplitudes. 

The original purpose of complex numbers was to solve certain algebraic equations with real coefficients. What is the corresponding purpose of negative probabilities? What plays the role of algebraic equations? These are the questions addressed in this paper. Our proposal is admittedly --- and provably --- limited. But it supports some of the usages of negative probabilities in the literature. One has to start somewhere.

For simplicity and to separate concerns, we restrict attention to finite spaces.

\section{Signed probability spaces}
\label{sec:prob}

We use, as a running example, a scenario due to Piponi \cite{Piponi}, which ``while artificial, is appealingly simple, and does convey some helpful intuitions'' \cite{AB}. 

\begin{example*}[Piponi's scenario]
A machine produces boxes with pairs $(l,r)$ of bits, each bit viewable through its own door. Somehow it is also possible to test whether the two bits are equal.
The probability of each possible combination of two bits is given by the following table:
\begin{align*}
&00 &&01 &&10 &&11\\
-&1/2 &&1/2 &&1/2 &&1/2
\end{align*}
While the table is not available to the observer, the following three experiments are available.
\begin{enumerate}
\item Look through the left door. This allows you to find out eventually that\\  
$\P(l=1)=1$ and $\P(l=0)=0$.
\item Look through the right door. This allows you to find out eventually that \\ 
$\P(r=1)=1$ and $\P(r=0)=0$.
\item Test whether the two bits are equal. This allows you to find out eventually that 
$\P(l\ne r) = 1$ and $\P(l=r) = 0$.
\end{enumerate}
Notice that these six discovered probabilities are all nonnegative (as in traditional probability theory), that they match the probabilities computed from the table above, but that they are not mutually consistent in traditional probability theory. 
\end{example*}

The following definition reflects our intent to work with finite spaces.

\begin{definition} By a \emph{signed probability space} $S$ we mean a pair $(\Omega,\P)$ where $\Omega$ is a nonempty set and $\P$ is a real-valued function on $2^\Omega$ such that the following probability laws hold.
\begin{enumerate}
\item[PL1.] $\P(\Omega)=1$.
\item[PL2.] If $e_1,e_2\subseteq\Omega$ and $e_1\cap e_2=\emptyset$ then $\P(e_1\cup e_2) = \P(e_1) + \P(e_2)$.
\end{enumerate}
If, in addition, we have
\begin{enumerate}
\item[PL3.] $\P(e)\ge0$ for all $e\subseteq\Omega$,
\end{enumerate}
then $\P$ and $S$ are \emph{traditional}. 
\hfill$\triangleleft$
\end{definition}

\begin{terminology*}
Here $\Omega$ is the \emph{sample space}, and its elements are \emph{sample points} or \emph{outcomes}. Subsets of $\Omega$ are \emph{events}. $\P$ is a \emph{signed probability distribution}. $\P(e)$, even if it is negative, is called the \emph{probability} of $e$. For brevity, when $\omega\in\Omega$, we write $\P(\omega)$ to mean $\P(\{\omega\})$.
\end{terminology*}

\begin{notation*} 
The complement $\Omega-e$ of an event $e$ will be denoted $\bar e$. The collection of all subsets of a set $s$ will be denoted $2^s$.  Disjoint union of sets $s_1,s_2$ will be denoted $s_1+s_2$. If $S$ is a set of sets then $\displaystyle \bigcup S = \bigcup_{s\in S} s$.
\end{notation*}

In the example above, the sample space $\Omega$ consists of the four binary strings $00, 01, 10, 11$,  and the probability distribution $\P$ is given by the table.

\begin{definition}
A \emph{test} for a signed probability space $S = (\Omega,\P)$ is given by (and, mathematically speaking, can be identified with) a partition of $\Omega$ into parts of nonnegative probability. 
\end{definition}

An \emph{execution} of a test picks out one of its parts. The example above explicitly exhibits three tests. 

\begin{quote}
\Q Normally, in traditional probability theory, an execution of a probability trial picks out an outcome $\omega$ with probability $\P(\omega)$. Why don't you do that in general? Require that only outcomes of nonnegative probability are picked out.

\A The proposed test is impossible in the case of nontraditional probability distribution. Notice that, in the example, the probability of the event $e^+ = \{\omega\in\Omega: \P(\omega)\ge0\}$ is more than 1. What would it mean to pick an outcome from $e^+$ according to a distribution with total probability $>1$?  Our definition of test intends to reflect measurements in quantum mechanics.
\end{quote}

\section{Observation frames}
\label{sec:frame}

Our goal in this section is to formalize the notions of an observable event and a coobservable set of events. 
Intuitively, a set $E$ of events is coobservable if there is a test $\tau_E$ that allows us to observe, for all $e\in E$, whether $e$ occurred or not. Further, $e$ is observable if the singleton set $\{e\}$ is coobservable. 

\begin{definition} An \emph{observation frame} is a pair $(\Omega,\CO)$, such that $\Omega$ is a nonempty set, CO is a collection of subsets of $\Omega$, and the following axioms hold:
\begin{itemize}
\item[CO1.]\ If $X\subseteq Y\in\CO$ then $X\in\CO$.
\item[CO2.]\ If $e_1,e_2\in X\in\CO$ then $X\cup\{\overline{e_1}\}\in\CO$ and $X\cup\{e_1\cup e_2\}\in\CO$. \hfill$\triangleleft$
\end{itemize}
\end{definition}

\begin{terminology*} Event sets in CO are \emph{coobservable}. An event $e$ is \emph{observable} if the set $\{e\}$ is coobservable. For brevity, maximal coobservable sets, maximal in the inclusion order, will be called \emph{ensembles}. \hfill$\triangleleft$
\end{terminology*}

\begin{notation*}
The set $\bigcup\CO$ of the observable events will be denoted Ob.
\end{notation*}

It follows from the definition that every ensemble is a Boolean algebra of subsets of $\Omega$. 

\begin{quote}
\Q How do you justify CO2.

\A Let $e_1,e_2\in X\in\CO$. Since $X$ is coobservable, there exists a test $\tau_X$ that allows us to observe, for all $e\in X$, whether $e$ occurred or not. Therefore $\tau_X$ also allows us to observe whether $\overline{e_1}$ occurred or not: it occurred if and only if $e_1$ didn't occur. And $\tau_X$ allows us to observe whether $e_1\cup e_2$ occurred or not: it occurred if and only if $e_1$ occurred or $e_2$ occurred.   
\end{quote}

\begin{proviso}[Finiteness]
By default, observation frames are finite, i.e., their sample spaces are finite. 
\end{proviso}

\begin{remark*}
The definition of observation frames should be more general by excluding $\Omega$ and dealing only with coobservation.  But, at this initial point of our investigation, we are willing to sacrifice the generality.
\hfill $\triangleleft$
\end{remark*}

Since every coobservable set is a subset of an ensemble and every subset of an ensemble is coobservable, the whole collection CO of coobservable event sets is determined by the ensembles.
 
Further, due to the finiteness proviso, the Boolean algebra of any ensemble $E$ is atomic. The atoms partition the sample space; let us call that partition $\Pi_E$. The partition $\Pi_E$ uniquely determines the ensemble $E$. 
Thus the collection CO can be given by the table of ensemble-induced partitions $\Pi_E$.

\begin{example*}[The observation frame of Piponi's scenario]
Piponi's scenario gives rise to the following observation frame. 
The sample space $\Omega$ consists of the four binary strings $00, 01, 10, 11$, and there are three ensembles giving rise to the following partitions:
\begin{align*}
&\big\{\{00,01\},\{10,11\}\big\} \\ 
&\big\{\{00,10\},\{01,11\}\big\}\\   
&\big\{\{00,11\},\{01,10\} \big\}
\end{align*}
\end{example*}

Furthermore, there is the least common refinement $\Pi$ of all the ensemble-induced partitions; it has the smallest number of parts. Notice that, for any part $e\in\Pi$, different outcomes in $e$ cannot be distinguished. For all practical purposes, members of $\Pi$ can be viewed as singletons.

\begin{proviso}[Fat outcomes] 
By default, each part of the common refinement contains a single outcome.
\end{proviso}

\section{Observation spaces}
\label{sec:space}

\begin{definition} An \emph{observation space} is a triple $(\Omega,\CO,\P)$, such that 
\begin{itemize}
\item $(\Omega,\CO)$ is an observation frame,
\item $\P$ is a real-valued function on $\Ob=\bigcup\CO$ satisfying the following versions of the probability laws PL1--PL3.
\begin{itemize}
\item $\P(\Omega)=1$.
\item If $(e_1,e_2)\in\CO$ and $e_1\cap e_2=\emptyset$ then $\P(e_1 + e_2) = \P(e_1) + \P(e_2)$. 
\item $\P(e)\ge0$ for all $e\in\Ob$.
\hfill$\triangleleft$
\end{itemize}
\end{itemize}
\end{definition}

Since each observable event belongs to some ensemble, it suffices to define $\P$ on every ensemble (and ensure that every event gets the same probability in every ensemble that contains it). Since any ensemble $E$ is a Boolean algebra of sets, the probability distribution on $E$ is determined by the probabilities assigned to the parts of the ensemble partition $\Pi_E$.

\begin{problem}[Extension] Given an observation space $(\Omega,\CO,\P)$,  do the following.
\begin{enumerate}
\item Decide whether there is a traditional probability distribution that extends $\P$ from observable events to all events.
\item If such a traditional extension exists then find one. 
\item Otherwise decide whether there is a signed probability distribution that extends $\P$ from observable events to all events.
\item If such a signed extension exists then find one.
\end{enumerate}
\end{problem}

The Fat-outcomes proviso of the previous section makes the Extension Problem trivial in the case of a single ensemble: the given $\P$ is already defined on all events. 
If there are exactly two ensembles $A$ and $B$, imposing partitions $\Pi_A = \{A_1, \dots, A_m\}$ and $\Pi_B = \{B_1, \dots, B_n\}$, and if every intersection $A_i\cap B_j\ne\emptyset$, then there is a traditional solution for the Extension Problem: set $\P(A_i\cap B_j) = \P(A_i)\cdot\P(B_j)$. 

\begin{example*}[The observation space of Piponi's scenario]
The observation frame of the scenario is described in the previous Example. It has exactly three ensembles. It remains only to specify the probability distribution $\P$. It is given by the following table.
\begin{align*}
\P\{00,01\}=0\quad &\P\{10,11\}=1\\
\P\{00,10\}=0\quad &\P\{01,11\}=1 \\
\P\{00,11\}=0\quad &\P\{01,10\}=1 
\end{align*}
\end{example*}

The following simple theorem will turn out to be useful.

\begin{theorem}[Symmetry]\label{thm:sym}
Let $G$ be a group of automorphisms of an observation space $O=(\Omega,\CO,\P)$. If $Q$ is a possibly-signed extension of $\P$ to all events, then so is the average
\[ R(e) = \frac1{|G|} \sum_{g\in G} Q(ge). \]
\end{theorem}

\begin{proof} If event $e$ is observable then $Q(ge)
= Q(e) = \P(e)$ for any  $g\in G$; hence $R(e) = \P(e)$. It remains to
check that $R$ is a possibly-signed probability distribution. Since
$\Omega$ is observable, $R(\Omega) = \frac1{|G|} \sum_{g\in G}
\P(\Omega) = \P(\Omega) = 1$.  Further,  if events $e_1,e_2$ are
disjoint and coobservable then
\[
R(e_1\cup e_2) 
= \frac1{|G|} \sum_{g\in G} Q(g(e_1\cup e_2))
 = \frac1{|G|} \sum_{g\in G} [Q(ge_1) + Q(ge_2)]
 = R(e_1) + R(e_2).
\]
\end{proof}

\section{Bell's theorem and negative probabilities}
\label{sec:bell}

Quantum theory is contextual in the sense that the value of an observable $O$,
measured as a part of one context, may differ from the value of $O$ measured as a
part of another context. Attempts to avoid contextuality may lead to negative probabilities. This will be illustrated in the present section. We start with two  Gedankenexperiments exhibiting the contextuality of quantum mechanics.

Prepare the state
$\frac1{\sqrt2} \Big(\ket{01} - \ket{10}\Big)$,
known as the spin singlet state, 
of a pair of spin 1/2 particles, e.g., electrons. 
Here \ket0 represents spin up in the $z$ direction, and \ket1 represents spin down. Choose an arbitrary direction $\bf a$ and measure spin in direction $\bf a$ on both particles, getting $+\frac12$ if the spin is up or $-\frac12$ if the spin is down. According to quantum mechanics, the results of the two measurements are opposite to one another: one measurement yields $+\frac12$ and the other  $-\frac12$ \cite[Box~2.7]{NC}. This is true even if each measurement is performed outside of the lightcone of the other and thus cannot possibly affect the other measurement. 

Alternatively, one can work with photons, which are spin 1 particles. In this connection, see Figure~1 which, together with the caption, is borrowed from \cite{Aspect}.

\vspace{4ex}\hspace{-4.2em}
\includegraphics[scale=0.93,trim=0 8.5in 0 0.3in,clip]{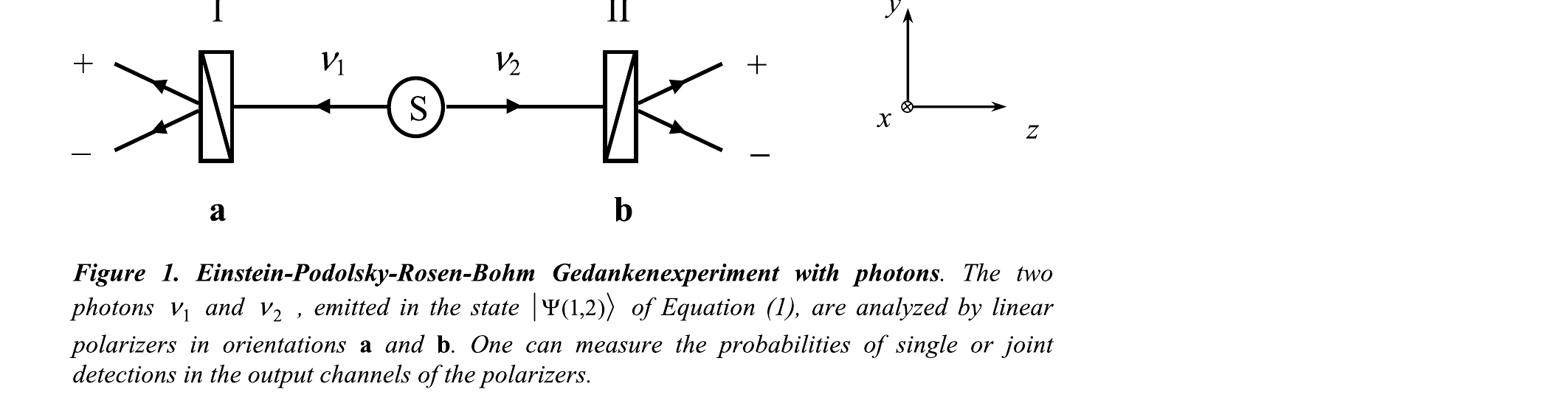}

\vspace{2ex}\noindent
The two photons are moving along the $z$ axis. Formula $\ket{\Psi(1,2)}$ in the figure is 
$\frac1{\sqrt2} \Big(\ket{00} + \ket{11}\Big)$ 
where \ket0 and \ket1 are unit vectors
in the Hilbert space for the quantum system of one photon. \ket0 and \ket1 correspond to polarization in the directions of the $x$ and $y$ axes respectively. If the orientations  $\bf a, \bf b$ of the analyzers in the figure coincide then the two measurement outcomes are guaranteed to coincide: both are $+1$ or both are $-1$.

Einstein, Podolsky and Rosen (EPR) saw contextuality as an indication that quantum mechanics is incomplete, i.e., that some hidden variables are missing \cite{EPR}. Bell famously proved that local hidden-variable theories contradict quantum-mechanical statistics \cite{Bell}. Subsequent experiments supported the latter. 

In general, the orientations $\bf a$, $\bf b$ of the two analyzers on Figure~1 may be different.
If $\theta$ is the angle $\angle(\bf a,\bf b)$ between vectors $\bf a$ and $\bf b$  then the two outcomes $(+1,+1)$ and $(-1,-1)$ in which the two measurements give us the same result, have probability $\frac12\cos^2\theta$ each, and the other two outcomes $(+1,-1)$ and $(-1,+1)$ have probability $\frac12\sin^2\theta$ each \cite[\S6-2]{Peres}. In particular, if $\bf a = \bf b$ then the probability of getting the same result is 1.

To illustrate how Bell's Theorem leads to negative probabilities, David Schneider played with three orientations $\vec A, \vec B, \vec C$ in his blog post \cite{Schneider}. The angle $\angle(\vec A, \vec C) = 3\pi/8$, and $\vec B$ is in between so that the angle $\angle(\vec A, \vec B) = \pi/4$ and $\angle(\vec B, \vec C) = \pi/8$. (We swapped Schneider's $\vec B$ and $\vec C$ so that our $\vec B$ is between $\vec A$ and $\vec C$.) He arrived at negative probabilities, implicitly assuming that (in terms of \S\ref{sec:space}) there is some solution of the appropriate Extension Problem. In the rest of this section, we explain Schneider's derivation and then address the implicit assumption.

Imagine that we work with a noncontextual hidden-variable theory where the measurements are determined locally, ``so the first thing we need to do is momentarily forget all our knowledge of quantum mechanics,'' \cite[page~148]{NC}. But we do use the data obtained in three optical experiments of the kind depicted in Figure~1.
\begin{enumerate}
\item Experiment $AB$ involves orientations $\vec A$ and $\vec B$,
\item Experiment $BC$ involves orientations $\vec B$ and $\vec C$, 
\item Experiment $AC$ involves orientations $\vec A$ and $\vec C$.
\end{enumerate}
This gives rise to the following observation space $T = (\Omega,\CO,\P)$.

\smallskip\noindent{\tt Sample space.}
The sample space $\Omega$ consists of eight sample points 
\begin{align*}
&1, a.k.a.\quad +++, && 2, a.k.a.\quad ++-, &&3, a.k.a.\quad +-+, &&4,\quad a.k.a. +--,\\
&5, a.k.a.\quad -++, && 6, a.k.a.\quad -+-, &&7, a.k.a.\quad --+, &&8,\quad a.k.a. ---.
\end{align*}
In the three-letter words $abc$ in alphabet $\{+,-\}$, the first letter $a$ is the result of measuring the spin in the $\vec A$ direction. Similarly for the second letter $b$ and the third letter $c$ using the $\vec B$ and $\vec C$ directions. 

\smallskip\noindent{\tt Coobservation.}
The first two letters $a,b$ of outcomes $abc$ give rise to an equivalence relation 
$ abc\equiv a'b'c'\iff (a=a'\land b=b')$ 
whose four equivalence classes 
$ \{1,2\}, \{3,4\}, \{5,6\}, \{7,8\} $
form a partition $\Pi_{12}$ of $\Omega$. Thanks to experiment $AB$, the four parts of $\Pi_{12}$ are coobservable.
Define partitions $\Pi_{23}$ and $\Pi_{31}$ similarly. The four parts of $\Pi_{23}$ are coobservable thanks to the experiment $BC$, and the four parts of $\Pi_{13}$ are coobservable thanks to the experiment $AC$.    

The four parts of any partition $\Pi_{ij}$ generate a Boolean algebra $E_{ij}$ of subsets of $\Omega$. Define
\[ \CO = \big\{ E: E\subseteq E_{ij}\text{ for some }i,j\big\}, \]
so that each $E_{ij}$ is an ensemble.

\smallskip\noindent{\tt Probability distribution.}
Since the experimental results support quantum mechanics, the experiments $AB$, $BC$ and $AC$ produce  results approximating the following three tables.

\begin{align*}\label{AB}
\P(++\pm) &= \P\{1,2\} = \frac12 \cos^2(\pi/4) = 1/4,\\
\P(--\pm) &= \P\{7,8\} = \frac12 \cos^2(\pi/4) = 1/4,\tag{AB}\\
\P(+-\pm) &= \P\{3,4\} = \frac12 \sin^2(\pi/4) = 1/4,\\
\P(-+\pm) &= \P\{5,6\} = \frac12 \sin^2(\pi/4) = 1/4.
\end{align*}

\begin{align*}\label{BC}
\P(\pm++) &= \P\{1,5\} = \frac12 \cos^2(\frac\pi8) = \frac18(2+\sqrt2),\\
\P(\pm--) &= \P\{4,8\} = \frac12 \cos^2(\frac\pi8) = \frac18(2+\sqrt2),
             \tag{BC}\\
\P(\pm+-) &= \P\{2,6\} = \frac12 \sin^2(\frac\pi8) = \frac18(2-\sqrt2),\\
\P(\pm-+) &= \P\{3,7\} = \frac12 \sin^2(\frac\pi8) = \frac18(2-\sqrt2).
\end{align*}

\begin{align*}\label{AC}
\P(+\pm+) &= \P\{1,3\} = \frac12 \cos^2(3\pi/8) = \frac18(2-\sqrt2),\\
\P(-\pm-) &= \P\{6,8\} = \frac12 \cos^2(3\pi/8) = \frac18(2-\sqrt2),
             \tag{AC}\\
\P(+\pm-) &= \P\{2,4\} = \frac12 \sin^2(3\pi/8) = \frac18(2+\sqrt2)\\
\P(-\pm+) &= \P\{5,7\} = \frac12 \sin^2(3\pi/8) = \frac18(2+\sqrt2).
\end{align*}

\medskip
Let $[a=b]$ be the event that the measurements for orientations $\vec A, \vec B$ coincide, and $[a\ne b]$ be the complementary event, that these measurements are distinct. Define events $[b=c], [b\ne c], [a=c]$ and $[a\ne c]$ similarly.
We have
\begin{align*}
[a=b] &= (++\pm)\cup(--\pm) = \{1,2,7,8\}, 
&&\P[a=b]= 1/2,\\
[a\ne b] &= (+-\pm)\cup(-+\pm) = \{3,4,5,6\}, 
&&\P[a\ne b]= 1/2,\\
[b=c] &= (\pm++)\cup(\pm--) = \{1,4,5,8\}, 
&&\P[b=c] = \cos^2(\pi/8) = \frac14(2+\sqrt2),\\
[a=c] &= (+\pm+)\cup(-\pm-) = \{1,3,6,8\}, 
&&\P[a=c]= \sin^2(\pi/8) = \frac14(2-\sqrt2).
\end{align*}

\smallskip
Suppose that a possibly-signed probability distribution $Q$ extends $\P$ to all events. Let $U$ be the nonobservable event $\{3,6\}$. The following computation shows that $Q$ cannot be traditional. 
\begin{align*}
2Q(U) &= \big[Q(U) + Q([a=c]-U)\big] + \big[Q(U) + Q([a\ne b]-U)\big] \\
&\qquad- \big[Q([a=c]-U) + Q([a\ne b]-U)\big]\\
&= \P[a=c] + \P[a\ne b] - \P[b=c] = \frac14(2-\sqrt2) +\frac12 -\frac14(2+\sqrt2)\\ &= \frac12(1-\sqrt2)\\
Q(U) &= \frac14(1-\sqrt2) < 0.
\end{align*}

Now let's address the question whether there is any solution of the Extension Problem in our case. 

Consider the transformation $g$ of $\Omega$ that, for any outcome $abc$, replaces every letter by its opposite. For example, $g(+-+)= -+-$. It is easy to see that $g$ is an automorphism of the observation space $T$. By Theorem~\ref{thm:sym}, the average $R(e) = \frac12(Q(ge) + Q(e))$ is a signed probability distribution that extends $\P$ to all events. 

	Since $Q\{3,6\} = \frac14(1-\sqrt2)$, we have:
\begin{align*}
&R(3) = R(6) = \frac12(Q(3)+Q(6)) = \frac12Q\{3,6\} = \frac18(1-\sqrt2)\\
&R(1) = R(8) = \frac18(2-\sqrt2) - \frac18(1-\sqrt2) = \frac18
  &&\text{by \eqref{AC}}\\
&R(2) = R(7) = \frac14 - \frac18 = \frac18 &&\text{by \eqref{AB}}\\
&R(4) = R(5) = \frac14 - \frac18(1-\sqrt2) = \frac18(1+\sqrt2) 
  &&\text{by \eqref{AB}}
\end{align*}

To prove that $R$ is consistent with $\P$, it suffices to check that these probabilities satisfy the constraints \eqref{AB}, \eqref{BC} and \eqref{AC} where $\P\{k,l\}$ is replaced with $R(k)+R(l)$. 

The \eqref{AB} constraints and the first two of the \eqref{AC} constraints are satisfied in a trivial way (because of the way they have been used to compute the outcome probabilities). $R(2)+R(4) = R(7)+R(5) = \frac18 + \frac18(1+\sqrt2) = \frac18(2+\sqrt2) =\P\{2,4\} = \P\{5,7\}$, and so the remaining two \eqref{AC} constraints are satisfied. We check the \eqref{BC} constraints.
\begin{align*}
&R(1)+R(5) = R(8)+R(4) = \frac18+\frac18(1+\sqrt2) = \frac18(2+\sqrt2) 
= \P\{1,5\} = \P\{4,8\}\\
&R(2)+R(6) = R(3)+R(7) = \frac18 + \frac18(1-\sqrt2) = \frac18(2-\sqrt2)
\end{align*}

\section{Hardy's Gedankenexperiment: Contextuality without negativity}
\label{sec:hardy}

The previous section may give one the idea that contextuality always leads to negative probabilities. In this section, building on Lucien Hardy's article \cite{Hardy} and also influenced by David Mermin's article \cite{Mermin}, we illustrate that this is not so. 

\begin{experiment*} 
Two one-qubit particles emerge from a common source heading for two far apart detectors. Aside from the passage of the particles from the source to the detectors, there are no connections between the source and either detector or between the two detectors.
The following picture is borrowed from \cite{Mermin} (and slightly modified). 

\hspace{-7em}
\includegraphics[scale=0.78,trim=0in 8.5in 0in 1.2in,clip]{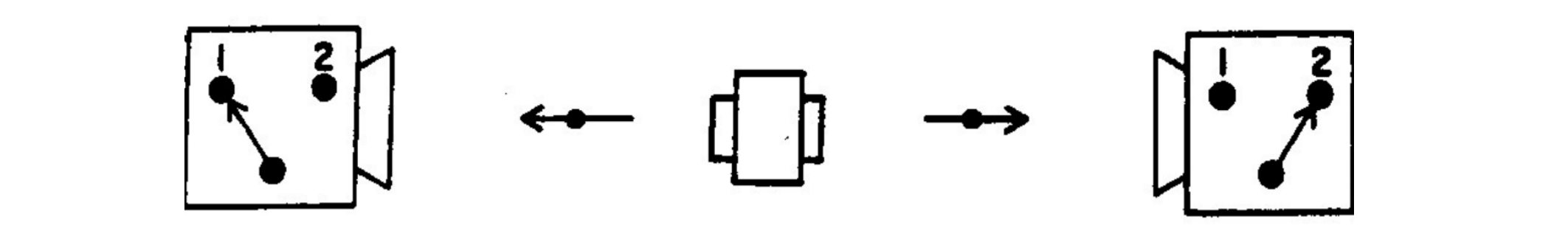}

\noindent
Each of the detectors is randomly set, ahead of time, to one of two modes, indicated by ``1'' and ``2'' in the picture. Four cases arise, two possible settings on each of the two detectors.

When a particle arrives at a detector, that detector performs a measurement and exhibits the result. In mode~1, observable $Z$ is measured. Its value is $+1$ in state \ket0 and $-1$ in state \ket1. In mode~2, observable $X$ is measured. Its value is $+1$ in state $\ket+ = \frac1{\sqrt2}(\ket0 + \ket1)$ and $-1$ in state $\ket- = \frac1{\sqrt2}(\ket0 - \ket1)$. Initially, the two particles are in state
\begin{equation}\label{1}
 \ket\psi = \frac1{\sqrt3}\Big(\ket{01}+\ket{10}-\ket{00}\Big).
\end{equation} 
That completes the description of the experiment. \hfill $\triangleleft$
\end{experiment*}

\bigskip
In the rest of the section, we analyse the experiment. The initial state is given to us in the basis \ket{00}, \ket{01}, \ket{10}, \ket{11}. It will be convenient to express it in three additional bases. In the basis \ket{0+}, \ket{0-}, \ket{1+},\ket{1-}, we have
\begin{equation}\label{2}
\sqrt3\ket\psi 
= -\ket0(\ket0-\ket1) + \ket1(\ket+ + \ket-)
= -\sqrt2\ket{0-} + \frac1{\sqrt2}\ket{1+} + \frac1{\sqrt2}\ket{1-}.
\end{equation}
In the basis \ket{+0}, \ket{+1}, \ket{-0}, \ket{-1}, we have
\begin{equation}\label{3}
 \sqrt3\ket\psi 
= \ket{01} - (\ket0-\ket1)\ket0 
= \frac{\left(\ket+ + \ket-\right)\ket1}{\sqrt2} - \sqrt2\ket{-0}
= \frac1{\sqrt2}\ket{+1} + \frac1{\sqrt2}\ket{-1} - \sqrt2\ket{-0}.
\end{equation}
In the basis \ket{++}, \ket{+-}, \ket{-+}, \ket{--}, we have
\begin{equation}\label{4}
\sqrt3\ket\psi 
= \frac12\ket{++} - \frac12\ket{+-} - \frac12\ket{-+} - \frac32\ket{--}.
\end{equation}

\smallskip\noindent{\tt Contextuality.}
Again, suppose that we work with a noncontextual hidden-variable theory where the measurements are determined locally. Let $f_l(Z)$ and $f_l(X)$ be the sets of values that may occur as the result of the $Z$ measurement and $X$ measurement respectively on the left. Define  $f_r(Z)$ and $f_r(X)$ similarly. 
By the noncontextuality assumption, these sets depend only on what happens on their side of the common source. But this leads to a contradiction. 

Indeed, consider the case $ZZ$, where both modes are 1 and thus observable $Z$ is measured on the left and the right. By \eqref{1}, the conditional probability $P\big[(+1,+1)\given ZZ\big]$ that we have $+1$ on the left and the right is $1/3$. Hence $+1\in f_l(Z)\cap f_r(Z)$. In particular, $+1\in f_l(Z)$. 

By \eqref{4}, in the case $XX$, where both modes are 2 and thus observable $X$ is measured on the left and the right, the conditional probability $\P\big[(+1,+1)\given XX\big]$ that $+1$ is produced on the left and the right is positive (namely $1/12$) and therefore $+1\in f_r(X)$. 

Now let's consider the case $ZX$, where the left mode is 1 and so $Z$ is measured on the left and where right mode is 2 and so $X$ is measured on the right. By noncontextuality, it must be possible to have $+1$ on the left and on the right in the same trial. But this does not happen. For, by \eqref{2}, the conditional probabilities in the $ZX$ case are as follows.
\[ \P\big[(+1,-1)\given ZX\big] = 2/3,\quad \P\big[(-1,+1)\given ZX\big]= 1/6,\quad  \P\big[(-1,-1)\given ZX\big] = 1/6, \]
so that $\P\big[(+1,+1)\given ZX\big]=0$ which gives the desired contradiction.

\smallskip\noindent{\tt Observation space.}
The sample space $\Omega$ consists of 16 outcomes for the combination of the random settings of modes and the observations of the results: $(ZZ,\pm1,\pm1)$, $(ZX,\pm1,\pm1)$, $(XZ,\pm1,\pm1)$ and $(XX,\pm1,\pm1)$. 

In the sense of the observation space, all 16 outcomes are observable. For example, consider the case $ZZ$ where both modes are 1. By \eqref{1}, the conditional probabilities $\P\big[(+1,+1)\given ZZ\big]$, $\P\big[(+1,-1)\given ZZ\big]$, $\P\big[(-1,+1)\given ZZ\big]$ are $1/3$. Accordingly 
\[ \P\big[\{(+1,+1),(+1,-1),(-1,+)\}\given ZZ\big] = 1 \]
and so $\P\big[(-1,-1)\given ZZ\big] = 0$. Thus, in the case $ZZ$, it is impossible to have $-1$ on the left and on the right in the same trial. But, like event $\emptyset$,  the event $\{(ZZ,-1,-1)\}$ is observable in the sense of \S\ref{sec:frame}. 

Thus $\P$ is defined on all events. The extension problem is trivial in Hardy's case, and negative probabilities do not arise.

\begin{remark*}
The presence of contextuality and the absence of negativity seem to contradict 
Robert Spekkens's claim that negativity and contextuality are equivalent forms of nonclassicality \cite{Spekkens}. Earlier, in a lengthy footnote in Section~6 of \cite{G236}, we showed that the equivalence claim is unsubstantiated.
\end{remark*}

\section{Limitation}

The notion of observation spaces was motivated by quantum mechanics with its observables, i.e., Hermitian operators, which may or may not be coobservable, i.e., commeasurable. Unfortunately, as this section shows, this notion is too simplistic to faithfuly model more complicated sets of quantum mechanical observables. 

Let $\cS$ be a set of Hermitian operators on a finite-dimensional
Hilbert space \H, and let $(\Omega,\CO)$ be an
observation frame. By a \emph{model} of $\cS$ in
$(\Omega,\CO)$ we mean a partial function $\mu$ from closed
subspaces of \H\ to subsets of $\Omega$ such that:
\begin{enumerate}
  \item[M1] For each operator $A\in\cS$ and each sum $E$ of eigenspaces
    of $A$, $\mu(E)$ is defined and is an observable event in
    $(\Omega,\mathrm{CO})$. 
  \item[M2] For each operator $A\in\cS$, the collection $\mathcal E_A$ of
    all sums of eigenspaces of $A$ has $\{\mu(E):E\in\mathcal E_A\}$
  coobservable in $(\Omega,\mathrm{CO})$.
\item[M3] If $E$ is a sum of  eigenspaces of some $A\in\cS$, then the
  complementary sum $E^\bot$ satisfies $\mu(E^\bot)=\Omega-\mu(E)$.
\item[M4] If $E_1$ and $E_2$ are sums of eigenspaces of a single $A\in\cS$, then $\mu(E_1+E_2)=\mu(E_1)\cup \mu(E_2)$.
\end{enumerate}

\begin{quote}
\Q Explain M1. Why are we talking about sums of eigenspaces?

\A Given a set $V$ of eigenvalues of a Hermitian operator $A\in\cS$, let $E$ be the sum of the corresponding eigenspaces of $A$. As a closed subspace of the Hilbert space, $E$ is an event. When we measure the Hermitian operator $A$, the result $v$ is one of its eigenvalues. By observing whether $v$ is in $V$ we know whether $E$ occurred. So $E$ should be observable in $(\Omega,\CO)$.

The same justification applies to M2.

\Q Actually M2 implies M1.

\A Almost. M1 contains the requirement that $\mu(E)$ is defined which is used implicitly in M2.
\end{quote}

\begin{theorem}\label{thm:limit}
  There is a finite set $\cS$ of Hermitian operators on $\C^4$
  that admits no model in any observation frame.
\end{theorem}

\begin{proof}
Consider the 18-vector, 9-basis example of Cabello et al.\ \cite{Cabello} given in color
on the Wikipedia page for ``Kochen-Specker Theorem'' \cite{W-KS}. For each of
the 9 bases $B$ there, invent an operator $S_B$ (Hermitian, with 4 distinct
eigenvalues) whose eigenspaces are exactly the 4 vectors in $B$ (and their
scalar multiples). Let $\cS$ be the set of these 9 $S_B$'s, and suppose, toward
a contradiction, that $\mu$ were a model of $\cS$ in $(\Omega,\CO)$.

Let $B$ be any one of the nine bases, and consider the sums $E$ of eigenspaces of $S_B$, i.e., the ``coordinate
subspaces'' of $\C^4$ with respect to the basis $B$. They
constitute a Boolean algebra $\BA(B)$ and, by the requirements for a
model, their $\mu$-images are defined and constitute a Boolean
subalgebra of $2^\Omega$; furthermore, when restricted to 
$\BA(B)$, $\mu$ is a Boolean homomorphism. 
The atoms of $\BA(B)$ are the eigenspaces themselves. The $\mu$-image of $\BA(B)$ is a finite Boolean algebra of events in $(\Omega,\CO)$. Its atoms are the $\mu$-images of the four eigenspaces. So each $\omega\in\Omega$ is in exactly one of those four $\mu$-images.  

Fix some point $\omega\in\Omega$. 
It selects, for each of the 9 bases $B$, exactly
one of the 4 vectors in $B$.  Furthermore, when the same vector occurs
in two bases, then it is selected from one if and only if it is selected from the
other. This is because for $\vec v$ to be selected means that (writing
$\langle \vec v\rangle$ for the subspace generated by $\vec v$)
$\omega\in \mu\langle \vec v\rangle$, and this doesn't depend on the
rest of a basis.

But the point of the example of Cabello et al.\ is precisely that such a
selection is impossible. Indeed, we would have 9 selections, one from each
basis, but every selection would occur twice because, as indicated by
the colors in \cite{W-KS}, every vector occurs in exactly two bases. So 9 would have
to be even.
\end{proof}

The $\C^4$ in the theorem can be improved to $\C^3$ at the cost of using a more complicated example, as in the original Kochen-Specker proof.

\begin{quote}
\Q This is too bad that there is such a brutal limitation on
observation spaces. And, since the theorem speaks about modeling in an observation frame, rather than observation space, negative probabilities are irrelevant.

\A It is not completely obvious to us at this point that negative probabilities are irrelevant. We can weaken conditions M3 and M4 by requiring that the equalities hold only up to an error of probability zero. The resulting weaker models would be defined not in observation frames but in
observation spaces because the weaker notion involves probabilities.  

If we stick to nonnegative probabilities, this doesn't buy us anything. But once negative probabilities enter the picture, errors of probability zero become more complicated. They may involve cancellation between
outcomes of positive probability and outcomes of negative probability. 

\Q I have another question. It seems that the proof of the limitative Theorem~\ref{thm:limit} crucially uses the existence of points in observation frames. Did you consider working more abstractly with just observable events and coobservable sets of events?

\A Yes, a little. We found that a ``pointless'' abstraction of observation frames seems to lead naturally to orthomodular lattices, extensively studied (though not by us) in quantum logic \cite{BvN}.
\end{quote}

\section*{Future work}

We introduced observation spaces and the extension problem. Observation spaces allow us to model some interesting quantum mechanical situations. And the extension problem amounts to asking when does one need negative probabilities and what can one accomplish with negative probabilities. This is a small step toward understanding what negative probabilites can be used for.  Observation spaces do not allow us to model straightforwardly the more complicated situations involved in proofs of the Kochen-Specker theorem. We'd like to understand what the next step should be.

\end{document}